\documentclass[aps,pra,twocolumn,nofootinbib, superscriptaddress,10pt]{revtex4-1}
\usepackage{graphicx}
\usepackage{epstopdf}
\usepackage{amsmath}
\usepackage{amssymb}
\usepackage{mathrsfs}
\usepackage{amsthm}
\usepackage{bm}
\usepackage{url}
\usepackage[T1]{fontenc}
\usepackage{csquotes}
\usepackage{color}
\MakeOuterQuote{"}
\newtheoremstyle{note}
  {\topsep/2}               % ABOVE SPACE
  {\topsep/2}               % BELOW SPACE
  {}                      % BODY FONT
  {\parindent}            % INDENT (empty value is the same as 0pt)
  {\itshape}              % HEAD FONT
  {.}                     % HEAD PUNCTUATION
  {5pt plus 1pt minus 1pt}% HEAD SPACE
  {}

\theoremstyle{note}
\newtheorem{theorem}{Theorem}
\newtheorem{lemma}{Lemma}

\newtheorem{corollary}{Corollary}
\newtheorem{proposition}{Proposition}

\theoremstyle{definition}

\theoremstyle{remark}
\newtheorem{remark}{Remark}

\newcommand{\tr}{\operatorname{tr}}

 \newcommand{\rme}{\mathrm{e}}
 \newcommand{\rmi}{\mathrm{i}}

 \newcommand{\rmB}{\mathrm{B}}

 \newcommand{\rmE}{\mathrm{E}}

 \newcommand{\rmT}{\mathrm{T}}

 \newcommand{\caB}{\mathcal{B}}
 
 \newcommand{\caH}{\mathcal{H}}

 \newcommand{\caR}{\mathcal{R}}
 \newcommand{\caS}{\mathcal{S}}

 \newcommand{\bbZ}{\mathbb{Z}}
  \newcommand{\nni}{\mathbb{Z}^{\geq0}}

  \newcommand{\rank}{\mathrm{rank}}

\newcommand{\be}{\begin{equation}}
\newcommand{\ee}{\end{equation}}
\newcommand{\ba}{\begin{align}}
\newcommand{\ea}{\end{align}}

\def\<{\langle}  %% overriding the original command \<
\def\>{\rangle}  %% overriding the original command \>

\newcommand{\id}{1}

\def\eqref#1{\textup{(\ref{#1})}}  %% overriding the original command \eqref
\newcommand{\eref}[1]{Eq.~\textup{(\ref{#1})}}

\newcommand{\esref}[1]{Eqs.~\textup{(\ref{#1})}}

\newcommand{\fref}[1]{Fig.~\ref{#1}}

\newcommand{\fsref}[1]{Figs.~\ref{#1}}

\newcommand{\tref}[1]{Table~\ref{#1}}

\newcommand{\sref}[1]{Sec.~\ref{#1}}

\newcommand{\ssref}[1]{Secs.~\ref{#1}}

\newcommand{\thref}[1]{Theorem~\ref{#1}}
\newcommand{\Thref}[1]{Theorem~\ref{#1}}

\newcommand{\lref}[1]{Lemma~\ref{#1}}
\newcommand{\Lref}[1]{Lemma~\ref{#1}}

\newcommand{\pref}[1]{Proposition~\ref{#1}}
\newcommand{\Pref}[1]{Proposition~\ref{#1}}
\newcommand{\psref}[1]{Propositions~\ref{#1}}
\newcommand{\Psref}[1]{Propositions~\ref{#1}}

\newcommand{\cref}[1]{Conjecture~\ref{#1}}
\newcommand{\Cref}[1]{Conjecture~\ref{#1}}

\newcommand{\aref}[1]{Appendix~\ref{#1}}

\newcommand{\rcite}[1]{Ref.~\cite{#1}}
\newcommand{\rscite}[1]{Refs.~\cite{#1}}

\begin{document}
\title{Optimal  verification and fidelity estimation of maximally entangled states}

\author{Huangjun Zhu}
\email{zhuhuangjun@fudan.edu.cn}

\affiliation{Department of Physics and Center for Field Theory and Particle Physics, Fudan University, Shanghai 200433, China}

\affiliation{State Key Laboratory of Surface Physics, Fudan University, Shanghai 200433, China}

\affiliation{Institute for Nanoelectronic Devices and Quantum Computing, Fudan University, Shanghai 200433, China}

\affiliation{Collaborative Innovation Center of Advanced Microstructures, Nanjing 210093, China}

\author{Masahito Hayashi}
\email{masahito@math.nagoya-u.ac.jp}
\affiliation{Graduate School of Mathematics, Nagoya University, Nagoya, 464-8602, Japan}

\affiliation{Shenzhen Institute for Quantum Science and Engineering,
	Southern University of Science and Technology,
	%No.1088 Xueyuan Avenue,Nanshan District,
	Shenzhen,
	518055, China}
\affiliation{Centre for Quantum Technologies, National University of Singapore, 3 Science Drive 2, 117542, Singapore}

\begin{abstract}
We study the verification of  maximally entangled states by virtue of the simplest measurement settings:  local projective measurements without adaption. We show that optimal protocols are in one-to-one correspondence with complex projective 2-designs constructed from orthonormal bases.  Optimal protocols with minimal  measurement settings are in one-to-one correspondence with complete sets of mutually unbiased bases. Based on this observation, optimal protocols are constructed explicitly for any local dimension, which can also be applied to estimating the fidelity with the target state and to detecting entanglement.  In addition, we show that incomplete sets of mutually unbiased bases are optimal for verifying maximally entangled states when the number of measurement settings is restricted. Moreover, we construct  optimal protocols for the adversarial scenario in which state preparation is not trusted. The number of tests has the same scaling behavior as the counterpart for the nonadversarial scenario; the overhead is no more than three times. 
We also show that the entanglement of the maximally entangled state can be certified with any given significance level using only one test as long as the local dimension is large enough.
\end{abstract}

\date{\today}
\maketitle

\section{Introduction}
Entanglement is a valuable resource in quantum information processing and a focus of foundational studies. Maximally entangled states are particularly useful because of their applications in many quantum information processing tasks, such as teleportation, dense coding, and quantum cryptography. They are also standard units in entanglement manipulations and transformations and thus play a fundamental role in the resource theory of entanglement \cite{HoroHHH09,GuhnT09}.

In practice, it is not easy to produce perfect maximally entangled states due to various experimental imperfections. It is therefore crucial to efficiently verify  the states produced within a given precision based on accessible measurements. Since it is in general very difficult to perform entangling measurements, it is natural to restrict our attention to measurements that can be realized by local operations and classical communication (LOCC)~\cite{HoroHHH09}, which is a standard paradigm in quantum information processing. This problem has been studied before \cite{HayaMT06,HayaSTM06,Haya09G,Haya08,HayaO17,PallLM18}, but most previous approaches entail continuous measurements, which are not practical; see \rcite{Haya08} for a preliminary study on the applications of discrete measurements based on symmetric informationally complete measurements \cite{FuchHS17} and mutually unbiased bases (MUB) \cite{DurtEBZ10,Ivan81,WootF89}.

In this work, we propose practical and efficient protocols for verifying maximally entangled states, which require only a few local projective (LP) measurements without adaption.  We  prove that optimal protocols based on LP measurements are in one-to-one correspondence with weighted complex projective 2-designs constructed from orthonormal bases.
Optimal protocols with minimal  measurement settings
 are in one-to-one correspondence with complete sets of MUB. For any local dimension $d$, optimal protocols can be constructed using at most $\lceil\frac{3}{4}(d-1)^2\rceil+1$ distinct measurement settings. 
These protocols can also be applied to fidelity estimation and entanglement detection. Besides,  incomplete sets of MUB can be used to construct optimal verification protocols when the number of measurement settings is restricted.

Moreover, our approach can be applied to the adversarial scenario in which the states to be verified are prepared by an untrusted party. In this case, our protocols built on LOCC are even optimal among protocols that can access entangling measurements. In addition, we prove that the entanglement of the maximally entangled state can be certified with any given significance level using only one test  as long as the local dimension is large enough. Again, entangling measurements are not necessary to achieve this goal. 
Compared with previous works on single-copy entanglement detection  \cite{DimiD18,SaggDGR19}, our protocol requires a smaller local dimension to achieve a given significance level.

 Our study not only provides practical and efficient  protocols for verifying maximally entangled states, but also highlights the operational significance of 2-designs and MUB. The connection between maximally entangled states and maximally incompatible measurements featured in this work is also of intrinsic interest.

\section{Verification of pure states}
\subsection{Nonadversarial scenario}
Before studying optimal verification of maximally entangle states under LOCC,
 it is instructive to review
the general framework of pure-state verification \cite{PallLM18}, though here we consider more general measurements. The main notations used in this paper are summarized in \tref{tab:notation}. 
Suppose we have  a device that is expected to prepare the target state $|\Psi\>$. In reality,  it turns out the device   produces $\sigma_1, \sigma_2, \ldots, \sigma_N$ in $N$ runs. Now our task is to distinguish between the two cases, assuming that either $\sigma_j$ is identical to $|\Psi\>\<\Psi|$ for all $j$ or satisfies the condition $\<\Psi|\sigma_j|\Psi\>\leq 1-\epsilon$ for all $j$ (this assumption can be relaxed thanks to \rscite{ZhuH19AdS,ZhuH19AdL}).
To this end we can perform two-outcome measurements from a set  of accessible measurements.  Each two-outcome measurement $\{E_j, 1-E_j\}$ is specified by a test operator $E_j$, which satisfies the condition $0\leq E_j\leq \id$ and corresponds to passing the test. It is natural to choose $E_j$ such that the target state $|\Psi\>$ always passes the test; that is, $E_j|\Psi\>=|\Psi\>$ for all $E_j$.
For comparison, the maximal probability that $\sigma_j$   can pass the test in the case $\<\Psi|\sigma_j|\Psi\>\leq 1-\epsilon$ is given by \cite{PallLM18,ZhuH19AdL}
\begin{equation}
\max_{\<\Psi|\sigma|\Psi\>\leq 1-\epsilon }\tr(\Omega \sigma)=1- [1-\beta(\Omega)]\epsilon=1- \nu(\Omega)\epsilon,
\end{equation}
where $\Omega=p_j E_j$ with $p_j$ being the probability of performing the test $E_j$,  $\beta(\Omega)$ is the second largest eigenvalue of $\Omega$, and  $\nu(\Omega)=1-\beta(\Omega)$ is the spectral gap from the maximal eigenvalue. Here $\Omega$ is referred to as a verification operator and a strategy.

After $N$ runs, $\sigma_j$ in the bad case can pass all tests with probability at most $[1-\nu(\Omega)\epsilon]^N$. To guarantee significance level $\delta$ (that is, $[1-\nu(\Omega)\epsilon]^N\leq \delta$), the minimum number of tests reads
\begin{equation}\label{eq:NumberTest}
N=\biggl\lceil\frac{ \ln \delta}{\ln[1-\nu(\Omega)\epsilon]}\biggr\rceil
\leq 
\biggl\lceil\frac{1}{\nu(\Omega)\epsilon}\ln\delta^{-1}\biggr\rceil. 
\end{equation}
Note that passing a single test 
guarantees the significance level 
$1-\nu(\Omega)\epsilon$.
That is, 
only one test is required if 
\begin{equation}\label{eq:OneTestCon}
\nu(\Omega)\epsilon+\delta\geq 1. 
\end{equation}
Here, the meaning of the significance level $\delta$ is that the probability of passing the test is not larger than $\delta$ as long as
$\<\Psi|\sigma_j|\Psi\>\leq 1-\epsilon$.

The number in \eref{eq:NumberTest} decreases monotonically with $\nu(\Omega)$. 
If all measurements are accessible, then the best strategy  is composed of the test $\{|\Psi\>\<\Psi|, 1-|\Psi\>\<\Psi|\}$ based on an entangling measurement. In this case we have $\Omega=|\Psi\>\<\Psi|$, $\nu(\Omega)=1$,  $N=\lceil\ln \delta/\ln(1-\epsilon)\rceil$, and the condition in  \eref{eq:OneTestCon} reduces to $\epsilon+\delta\geq 1$.

\begin{table}[htpb]
  \caption{\label{tab:notation}Summary of notations}
\begin{center}
  \begin{tabular}{|c|l|} 
  	  \hline
  	$d$ ($D$) &Dimension of the local (whole) Hilbert space\\
  \hline
 $\Psi$ & Generic target state\\
   \hline
 $\Phi$ & Maximally entangled  state\\
  \hline
 $\epsilon$ & Infidelity $1- \langle\Psi|\sigma|\Psi \rangle $ when $\sigma$ is the true state\\
  \hline
 $\delta$ & Significance level \\
\hline
 $\Omega$ & Strategy (verification operator) for one-copy state \\
\hline
 $\beta(\Omega)$ & Second largest eigenvalue of $\Omega$ \\
\hline
 $\nu(\Omega)$ & Spectral gap of $\Omega$, that is, $1-\beta(\Omega)$\\
\hline
 $N$ & Number of tests \\
\hline
  \end{tabular}
\end{center}
\end{table}

\subsection{Adversarial scenario}
In the adversarial scenario, the states are prepared by a potentially malicious adversary. In this case, we can still verify the target state by  first performing a random permutation on $N+1$ systems before applying  the  strategy $\Omega$ to $N$ systems  \cite{ZhuH19AdS,ZhuH19AdL}.  Now the performance depends on other eigenvalues of $\Omega$ in addition to $\beta(\Omega)$ (or $\nu(\Omega)$). Denote by  $F(N,\delta,\Omega)$ the minimum fidelity of the reduced state of the remaining party with the target state when $N$ tests are passed with significance level $\delta$. Denote by $N(\epsilon,\delta,\Omega)$  the minimum number of tests required to verify the target state within infidelity $\epsilon$ and significance level $\delta$. 
In general, it is not easy to derive analytical formulas for $F(N,\delta,\Omega)$ and $N(\epsilon,\delta,\Omega)$. Nevertheless,  such formulas have been derived in \rscite{ZhuH19AdS,ZhuH19AdL} for two cases most relevant to the current study.

According to \rscite{ZhuH19AdS,ZhuH19AdL}, if $\Omega$ is singular (has a zero eigenvalue), then we have 
\begin{equation}\label{eq:FiddelAdvUB}
F(N,\delta,\Omega)\leq 1-\min\left\{\frac{1-\delta}{N\delta\nu(\Omega)},\;\frac{1}{(N+1)\delta},\;1\right\}.
\end{equation}
If in addition $\nu(\Omega)\geq 1/2$, then the upper bound is saturated; that is,
\begin{equation}\label{eq:FiddelAdv}
F(N,\delta,\Omega)=1-\min\left\{\frac{1-\delta}{N\delta\nu(\Omega)},\;\frac{1}{(N+1)\delta},\;1\right\}.
\end{equation}
The minimum number of tests required to verify $|\Psi\>$ within infidelity $\epsilon$ and significance level $\delta$ reads \cite{ZhuH19AdL}
\begin{equation}\label{eq:NumberTestAdv}
N(\epsilon,\delta,\Omega)= \min\left\{
\biggl\lceil\frac{1- \delta}{\nu(\Omega) \delta \epsilon}\biggr\rceil,\; \biggl\lceil\frac{1}{\delta\epsilon}-1\biggr\rceil\right\}.
\end{equation}
Compared with \eref{eq:NumberTest}, the scaling  with $\delta$ in \eref{eq:NumberTestAdv} is suboptimal. When the strategy $\Omega$ is composed of the entangling test $\{|\Psi\>\<\Psi|, 1-|\Psi\>\<\Psi|\}$ for example, we have $\nu(\Omega)=1$, and $N(\epsilon,\delta,\Omega)$ is minimized among singular verification strategies. In this case,
\esref{eq:FiddelAdv} and \eqref{eq:NumberTestAdv} reduce to 
\begin{align}
F(N,\delta,\Omega)&=\max\left\{\frac{(N+1)\delta-1}{N\delta},0\right\},  \label{eq:FiddelAdvEnt} \\
N(\epsilon,\delta,\Omega)&= 
\biggl\lceil\frac{1- \delta}{\delta \epsilon}\biggr\rceil. \label{eq:NumberTestAdvEnt}
\end{align}
Therefore, it is impossible to verify the target state within infidelity $\epsilon<1$ and significance level $\delta\leq 1/2$ using only one test for any singular strategy \cite{ZhuH19AdL}.

For a given $\beta(\Omega)$, the optimal performance is achieved when the strategy  $\Omega$ is \emph{homogeneous}  \cite{ZhuH19AdS, ZhuH19AdL}, which means  it has the form
\begin{equation}
\Omega=|\Psi\>\<\Psi|+\lambda(\id-|\Psi\>\<\Psi|),
\end{equation}
where $\lambda=\beta(\Omega)$. In this case, it is natural to write $F(N,\delta,\lambda)$ and $N(\epsilon,\delta,\lambda)$ in place of  $F(N,\delta,\Omega)$ and $N(\epsilon,\delta,\Omega)$. The efficiency of $\Omega$ is determined by \esref{eq:FiddelAdvEnt} and \eqref{eq:NumberTestAdvEnt} when $\lambda=0$. When $0<\lambda<1$, 
define
\begin{align}
\eta_k(\lambda)&:=\frac{k\lambda^{k-1}+(N+1-k)\lambda^k}{N+1},\\ \zeta_k(\lambda)&:=\frac{(N+1-k)\lambda^k}{N+1}.
\end{align}
Then $F(N,\delta,\Omega)$ is given by \cite{ZhuH19AdS, ZhuH19AdL}
\begin{equation}\label{eq:MinFidelityAdv}
F(N,\delta,\lambda)=\begin{cases}
0, & \delta\leq \lambda^N,\\
\delta^{-1}[p_1 \zeta_{k}(\lambda)+p_2\zeta_{k+1}(\lambda)], &\mbox{otherwise}.
\end{cases}
\end{equation}
Here $k$ is the largest integer  that satisfies $\eta_k(\lambda)\geq \delta$, and $p_1,p_2$  are probabilities determined by the conditions 
\begin{equation}
p_1+p_2=1,\quad p_1 \eta_{k}(\lambda)+p_2\eta_{k+1}(\lambda)=\delta. 
\end{equation}

When $N=1$, \eref{eq:MinFidelityAdv} reduces to (cf. Proposition 2 in \rcite{ZhuH19AdL})
\begin{equation}\label{eq:MinFidelityAdvN1}
F(N,\delta,\lambda)=\begin{cases}
0, & \delta\leq \lambda,\\[0.2ex]
\frac{\lambda(\delta-\lambda)}{\delta(1-\lambda)}, & \lambda\leq \delta\leq \frac{1+\lambda}{2}, \\[0.6ex]
\frac{\delta(2-\lambda)-1}{\delta(1-\lambda)}, &  \frac{1+\lambda}{2}\leq \delta\leq 1.
\end{cases}
\end{equation}
Therefore, one test suffices to verify the target state within infidelity $\epsilon$ and significance level $\delta$ (assuming $0<\epsilon,\delta<1$) as long as 
\begin{equation}\label{eq:OneTestConAdv}
\frac{\lambda(\delta-\lambda)}{\delta(1-\lambda)}\geq 1-\epsilon. 
\end{equation}
This  condition is also necessary when $\delta\leq (1+\lambda)/2$.
Note that 
the requirement $\delta\geq \lambda$ is implicitly  implied in  \eref{eq:OneTestConAdv}. 
As an implication, passing a single test can guarantee significance level
\begin{equation}\label{eq:SLoneTestAdv}
\delta\geq \frac{\lambda^2}{\lambda-(1-\lambda)(1-\epsilon)}. 
\end{equation}

Suppose $0<\epsilon,\delta, \lambda<1$. Then  the minimum number $N(\epsilon,\delta,\lambda)$ of tests required to verify $|\Psi\>$ within infidelity $\epsilon$ and significance level $\delta$ in the adversarial scenario is given by \cite{ZhuH19AdS, ZhuH19AdL}
\begin{align}
N(\epsilon, \delta,\lambda)&=\Bigl\lceil\min_{k\in \nni}\!\tilde{N}(\epsilon,\delta,\lambda,k)\Bigr\rceil=\bigl\lceil\tilde{N}(\epsilon,\delta,\lambda,k^*)\bigr\rceil, \label{eq:NumTestHomo1}
	\end{align}
	where 
	\begin{align}
	\tilde{N}(\epsilon,\delta,\lambda,k)&:=\frac{k\nu^2 \delta F +\lambda^{k+1}+\lambda\delta(k\nu-1)}{\lambda\nu\delta \epsilon} \label{eq:NumTestHomoC}
	\end{align}
	with $F=1-\epsilon$ and $\nu=1-\lambda$. Here $\nni$ denotes the set of nonnegative integers, and $k^*$ is the largest integer $k$ that satisfies $\delta\leq \lambda^{k}/(F\nu+\lambda )= \lambda^{k}/(F+\lambda \epsilon)$. In addition, $k^*$ is equal to either $k_-:=\lfloor \log_\lambda \delta \rfloor$ or  $k_+:=\lceil \log_\lambda \delta \rceil$.
	In the limit $\delta\rightarrow 0$, the number  $N(\epsilon,\delta,\lambda)$ can be approximated as follows (assuming that $\lambda$ is lower bounded by a positive constant) \cite{ZhuH19AdL}, 
\begin{equation}\label{eq:NumberTestAdvLS}
N(\epsilon,\delta,\lambda)\approx\frac{F+\lambda\epsilon}{\lambda\epsilon\ln\lambda}
\ln\delta,
\end{equation}
where $F=1-\epsilon$. 
In the high-precision limit $\epsilon,\delta\rightarrow 0$, which is the situation of the most interest, we have
\begin{equation}\label{eq:NumberTestAdvOpt}
N(\epsilon,\delta,\lambda)\approx(\lambda\epsilon\ln\lambda)^{-1}
\ln\delta \geq\rme\epsilon^{-1}
\ln\delta^{-1},
\end{equation}
where the inequality is saturated 
at $\lambda=1/\rme$,  with $\rme$ being the base of the natural logarithm. Here  the number of required tests has the same scaling behaviors with $\epsilon$ and $\delta$ as in the nonadversarial scenario, and  the efficiency of the homogeneous strategy is characterized  by the function $1/(\lambda\ln\lambda^{-1})$. The optimal performance is achieved when $\lambda=1/\rme$,  in which case the overhead is only $\rme$ times.

\begin{figure}
	\includegraphics[width=6cm]{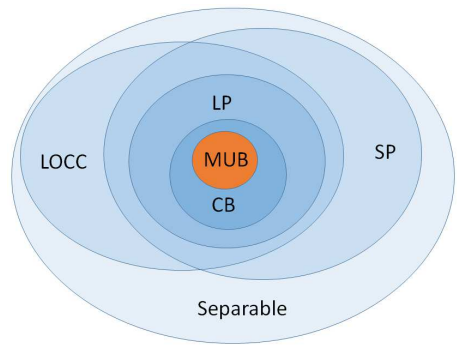}	
	\caption{\label{fig:HVS}Hierarchy of separable verification strategies for maximally entangled states. LP, local projective; SP, separable projective; CB, conjugate basis. Optimal CB strategies can achieve the same performance as the optimal separable strategies. Optimal MUB strategies can also achieve the same performance when the local dimension is a prime power. }
\end{figure}

\section{Limitations of local operations and classical communication}
\subsection{General discussions}

To analyze the limitations of local measurements on quantum state verification, we need to introduce several additional concepts.
A test operator $E$  is \emph{separable} if it is a linear combination of pure product states with nonnegative coefficients. The test $\{E,\id-E\}$ is separable if both $E$ and $\id-E$ are separable. A verification strategy is separable if it is composed of separable tests.  The strategy is \emph{separable projective} (SP) if, in addition, each test operator is a projector. Any verification strategy realized by LOCC is separable, as illustrated in \fref{fig:HVS}.

The robustness of entanglement of a quantum state $\rho$  \cite{HoroHHH09,VidaT99,HarrN03,Stei03,Bran05} is defined as
\begin{align}
E_\caR(\rho)&:=\min\Bigl\{x \Big|x\geq0, \, \exists \mbox{ a state } \sigma,\, \frac{\rho+x\sigma}{1+x}\in \caS \Bigr\}, \label{eq:RoE}
\end{align}
where $\caS$ denotes the set of separable states. If $\sigma$ is required to be the completely mixed state, we get the random robustness \cite{VidaT99},
\begin{align}
R(\rho)&:=\min\Bigl\{x \Big|x\geq0, \, \frac{D\rho+x}{D(1+x)}\in \caS \Bigr\}, \label{eq:RoE2}
\end{align}
where $D$ is the dimension of the whole Hilbert space.
Given a pure state $|\Psi\>$,  the following quantity is closely related to the robustness of entanglement,
\begin{align}
T(\Psi)&:=\min\Bigl\{\tr(E)
\Big| E \ge |\Psi\>\<\Psi| \Bigr\},
\label{eq:RoE3}
\end{align}
where the minimization is taken over
separable tests of the form $\{E,\id-E\}$.
The following lemma is an easy consequence of the definitions of $E_\caR(\Psi)$ and $T(\Psi)$. 
\begin{lemma}\label{lem:OptBoundB}[\cite[Theorem~2]{HayaMMO06} and \cite{OwarH08}]
Any pure state $|\Psi\>$ satisfies
\begin{align} 
T(\Psi)\ge 
E_\caR(\Psi) +1. 
\end{align}
\end{lemma}

Each test operator $E$ of a separable strategy $\Omega$  for $|\Psi\>$ satisfies the inequality
$\tr(E)\geq T(\Psi)\geq E_\caR(\Psi)+1$, so we have 
$\tr(\Omega)\geq T(\Psi)\geq E_\caR(\Psi)+1$. 
Similarly, we have $\tr(\Omega)\geq [1+R(\Psi)]/[1+R(\Psi)/D]$ if $\Omega$ is homogeneous. This observation implies the following lemma
given that $|\Psi\>$ is an eigenstate of $\Omega$  with eigenvalue~1.
\begin{lemma}\label{lem:OptBound}
Any separable strategy $\Omega$  for $|\Psi\>$ satisfies 
\begin{equation}
\beta(\Omega)
\geq\frac{ T(\Psi)-1}{D-1}\geq \frac{E_\caR(\Psi)}{D-1},
\end{equation}
where $D$ is the dimension of the whole Hilbert space. If $\Omega$ is homogeneous, then 
\begin{equation}
\beta(\Omega)\geq
 \frac{R(\Psi)}{D+R(\Psi)}.
\end{equation}
\end{lemma}

\begin{lemma}\label{lem:OptOrth}
	The verification operator $\Omega$ of any SP strategy  $\{P_l, p_l\}_{j=1}^g$ for an entangled state $|\Psi\>$ satisfies $\beta(\Omega) \geq 1/g$. 
	The bound is saturated if and only if (iff)
	$\bar{P}_l:=P_l-|\Psi\>\<\Psi|$ are mutually orthogonal and	$p_l=1/g$ for all $l$.
\end{lemma}
\begin{proof}
Since $|\Psi\>$ is entangled,  all projectors $P_l$  have ranks at least 2; that is, $\bar{P}_l$ have ranks at least  1. Therefore, 
\begin{equation}
	\beta(\Omega) = \biggl\|\sum_lp_l \bar{P}_l\biggr\|\geq \max_l p_l\geq \frac{1}{g}.
	\end{equation}
Here the second inequality is saturated iff $p_l=1/g$ for all $l$. In that case, the first inequality is saturated iff $\bar{P}_l$ are mutually orthogonal  for all $l$. This observation completes the proof of \lref{lem:OptOrth}. 
\end{proof}

An  SP strategy $\Omega$ composed of $g$ distinct tests
is   \emph{parsimonious} if $\beta(\Omega) = 1/g$; that is, $\nu(\Omega)=(g-1)/g$. In this case, by \eref{eq:NumberTest},
 the number of tests required to verify  $|\Psi\>$ within infidelity $\epsilon$ and significance level $\delta$ reads
\begin{equation}\label{eq:ParsiBoundN}
N=\biggl\lceil\frac{ \ln \delta}{\ln[1-(g-1)g^{-1}\epsilon]}\biggr\rceil\leq \biggl\lceil\frac{g}{(g-1)\epsilon}\ln\delta^{-1}\biggr\rceil.
\end{equation}
The counterpart for the adversarial scenario is given by  \eref{eq:NumberTestAdv} with $\nu(\Omega)=(g-1)/g$, assuming that $\Omega$ is singular.

\subsection{Bipartite pure states}
Now we turn to a bipartite system with the Hilbert space $\caH\otimes \caH$ of dimension $D=d^2$. Up to a local unitary transformation, any bipartite pure state can be expressed as 
\begin{equation}
|\Psi\>=\sum_{j=0}^{d-1} s_j|jj\>,
\end{equation}
where the Schmidt coefficients $s_0\geq s_1\geq \cdots\geq  s_{d-1}$ are arranged in decreasing order and satisfy the normalization condition $\sum_j s_j^2=1$. 
The robustness and random robustness of entanglement  of $|\Psi\>$ are well known \cite{HoroHHH09,VidaT99,HarrN03,Stei03,Bran05}, as reproduced here,
\begin{equation}\label{eq:RoEbip}
E_\caR(\Psi)=\biggl(\sum_j s_j\biggr)^2-1,\quad R(\Psi)=Ds_0s_1. 
\end{equation}
In addition, Theorem~2 of \rcite{OwarH08} showed that
\begin{equation}\label{eq:RoEbip1}
T(\Psi)=\biggl(\sum_j s_j\biggr)^2.
\end{equation}

By \lref{lem:OptBound},  any separable verification strategy $\Omega$ for $|\Psi\>$ satisfies 
 \begin{equation}\label{eq:betaBip}
\beta(\Omega)\geq\frac{\bigl(\sum_j s_j\bigr)^2-1}{d^2-1}.
 \end{equation}
If $\Omega$ is homogeneous, 
then we have a stronger conclusion,
\begin{equation}
\beta(\Omega)\geq\frac{s_0s_1}{1+s_0s_1}.
\label{eq:betaBipHomo}
\end{equation}
Here  the  inequality also follows from the fact that the partial transpose $(|\Psi\>\<\Psi|)^{\rmT_\rmB}$ has an eigenvalue equal to $-s_0s_1$, while a separable verification operator is necessarily positive partial transpose (PPT).

\subsection{Maximally entangled states}
We are particularly interested in maximally entangled states, which have the form  
\begin{equation}
|\Phi\>=\frac{1}{\sqrt{d}}\sum_j |jj\>
\end{equation}
up to local unitary transformations. 
According to \eref{eq:betaBip} or \eref{eq:betaBipHomo}, any separable strategy $\Omega$ for $|\Phi\>$ satisfies
\begin{equation}\label{eq:betanuBounds}
\beta(\Omega)\geq\frac{1}{d+1},\quad \nu(\Omega)\leq \frac{d}{d+1}.
\end{equation}
For a homogeneous strategy $\Omega$,
Theorem~1 of \rcite{HayaMT06} showed that
\begin{equation}
\Omega \ge \frac{1+d|\Phi\>\<\Phi|}{d+1}.\label{H1}
\end{equation}
Also, \rcite{HayaMT06} showed the existence of a local strategy that  saturates the inequality in 
\eref{H1}.
The bounds in \eref{eq:betanuBounds}
 are saturated iff the inequality  in \eref{H1} is saturated,  in which case we have
 \begin{equation}\label{eq:OptOmega}
\Omega=\frac{1+d|\Phi\>\<\Phi|}{d+1}.
\end{equation}
A separable strategy $\Omega$ for $|\Phi\>$ is \emph{optimal} if it saturates the bound $\beta(\Omega)\geq 1/(d+1)$ or $\nu(\Omega)\leq d/(d+1)$ in \eref{eq:betanuBounds}; an SP strategy is \emph{perfect} if it is both optimal and parsimonious.

For the optimal strategy, \eref{eq:NumberTest} reduces to 
\begin{equation}\label{eq:OptimalBoundN}
N=\biggl\lceil\frac{ \ln \delta}{\ln[1-d(d+1)^{-1}\epsilon]}\biggr\rceil\leq
\biggl\lceil\frac{d+1}{d\epsilon}\ln\delta^{-1}\biggr\rceil. 
\end{equation}
In the independent and identically distributed (i.i.d.) case, this result can  be derived 
from Sec.~4.3.2 and Eq.~(36) of \rcite{Haya09G}; see  \rcite{PallLM18} for the case $d=2$. 
Notably, thanks
 to \eref{eq:OneTestCon},
passing a single test 
guarantees the significance level 
$1-\frac{d}{d+1}\epsilon$
in the nonadversarial scenario.
The counterpart for the adversarial scenario is given by \eref{eq:SLoneTestAdv} with $\lambda=1/(d+1)$. 
These conclusions will have  important implications for entanglement detection as we shall see in \ssref{sec:EntDetectNA} and \ref{sec:EntDetectAdv}. In the large-$d$ limit, we have
 \begin{equation}
N=\biggl\lfloor \frac{\ln \delta}{\ln(1-\epsilon)}\biggr\rfloor+1\leq
\biggl\lceil\frac{\ln\delta^{-1}}{\epsilon}\biggr\rceil.
\end{equation} 
The number of tests is almost the same as what is required by the best
strategy based on entangling measurements.

Let $\rho$ be an arbitrary quantum state and
\begin{equation}
F(\rho,|\Phi\>\<\Phi|):=\tr(\rho |\Phi\>\<\Phi| )
\end{equation}
the fidelity between $\rho$ and $|\Phi\>\<\Phi|$.
If $\Omega$ is the optimal strategy given in \eref{eq:OptOmega}, then 
\begin{equation}
\tr(\rho\Omega)=\frac{1+dF(\rho,|\Phi\>\<\Phi| )}{d+1};
\end{equation}
(cf.~Theorem~1 in  \rcite{HayaMT06}),
so the fidelity can be inferred  from the passing probability,
\begin{equation}\label{eq:FidelityEst}
F(\rho,|\Phi\>\<\Phi|)=\frac{(d+1)\tr(\rho\Omega)-1}{d}.
\end{equation}
Therefore, homogeneous verification strategies can also serve for fidelity estimation.

\section{Verification of maximally entangled states}
\subsection{\label{sec:OptNA}Parsimonious and optimal verification strategies}
Given any  basis $\caB=\{|\psi_1\>, |\psi_2\>, \ldots, |\psi_d\>\}$ for $\caH$ (in this paper we only consider orthonormal bases), we can devise a \emph{conjugate-basis} (CB) test as follows: Alice performs the projective measurement on the basis $\caB$, while Bob performs the projective measurement on the conjugate basis $\caB^*=\{|\psi^*\>:|\psi\>\in \caB\}$, where $|\psi^*\>$ denotes the complex conjugate of $|\psi^*\>$ (with respect to the  given computational basis used to define $|\Phi\>$). The CB test is passed if Alice and Bob obtain the same outcome; in other words, the pass eigenspace is spanned by $|\psi\>\otimes |\psi^*\>$ for all $|\psi\>\in \caB$, and the test projector has the form
\begin{equation}\label{eq:CBproj}
P(\caB):=\sum_{|\psi\>\in \caB}
 |\psi\> \< \psi|  \otimes |\psi^*\>  \<\psi^*|,
\end{equation}
which has rank $d$.
A similar idea was used  to construct general tests from positive operator-valued measures (POVMs) on $\caH$ (see Eq.~(13) of \rcite{Haya09G}), while the test here is simpler and easier to realize.  Note that $P(\caB)|\Phi\>=|\Phi\>$ for any orthonormal basis $\caB$ of $\caH$, so $|\Phi\>$ can pass the test with certainty as desired.  A CB strategy $\{P(\caB_l),p_l\}_l$ is composed of CB tests, where $\caB_l$ are bases for $\caH$, and $p_l$ form a probability distribution. The resulting verification operator reads 
\begin{equation}
\Omega=\sum_l p_l P_l=\sum_{l}p_l \sum_{|\psi\>\in \caB_l}
|\psi\> \< \psi|  \otimes |\psi^*\>  \<\psi^*|.
\end{equation}

\begin{proposition}\label{pro:OrthMUB}
Let $\caB_1$ and $\caB_2$ be two orthonormal bases for $\caH$. Then $\bar{P}(\caB_1)$ and $\bar{P}(\caB_2)$ are orthogonal iff $\caB_1$ and $\caB_2$ are mutually unbiased. 
\end{proposition}
Here $\bar{P}(\caB)={P}(\caB)-|\Phi\>\<\Phi|$. 
 \Pref{pro:OrthMUB} is an implication of the following inequality
\begin{align}
&\tr[\bar{P}(\caB_1)\bar{P}(\caB_2)]+1=\tr[P(\caB_1)P(\caB_2)]\nonumber\\
&=\sum_{|\psi\>\in \caB_1,|\varphi\>\in \caB_2} |\<\psi|\varphi\>|^4\geq 1,
\end{align}
which is saturated iff $\caB_1$ and $\caB_2$ are mutually unbiased; that is, $|\<\psi|\varphi\>|^2=1/d$ for all $|\psi\>\in \caB_1$ and $|\varphi\>\in \caB_2$ \cite{DurtEBZ10}.  \Lref{lem:OptOrth} and \pref{pro:OrthMUB} together yield the following proposition. 
\begin{proposition}\label{pro:CBparsimonious}
Let  $\caB_1, \caB_2,\ldots, \caB_g$ be $g$ bases for $\caH$. The CB strategy  $\{P(\caB_l),p_l\}$ is parsimonious iff $\caB_1, \caB_2,\ldots, \caB_g$ are mutually unbiased and all $p_l$ are equal to $1/g$. The strategy is perfect iff, in addition, $g=d+1$, so that $\caB_1, \caB_2,\ldots, \caB_g$ form a complete set of MUB.
\end{proposition}

An MUB strategy is 
a CB strategy based on MUB and with uniform probabilities, which  is parsimonious by \pref{pro:CBparsimonious}. 
If a set of   $g$ MUB  is available, then the number of tests required  to verify  $|\Phi\>$ within infidelity $\epsilon$ and significance level $\delta$ in the  nonadversarial  scenario
is given by \eref{eq:NumberTest}  with $\nu(\Omega)=(g-1)/g$; that is, $\beta(\Omega)=1/g$. 
Incidentally, the maximally entangled state $|\Phi\>$ is equivalent to a qudit graph state for which we have introduced a general verification protocol called  
 the cover protocol  \cite{ZhuH19E}; see also \rcite{HayaM15} when $d=2$. In retrospect, the cover protocol in  this special case is equivalent to an MUB strategy constructed from two bases (that is $g=2$).

For $d\geq2$, there exist at least three bases that are mutually unbiased \cite{DurtEBZ10}. Define operators $Z$ and $X$ as follows,
\begin{equation}\label{eq:ZX}
Z|j\>=\omega^j |j\>,\quad X |j\>=|j+1\>,\quad \omega=\rme^{2\pi\rmi/d},
\end{equation}
where $j\in \bbZ_d$ and $\bbZ_d$ is the ring of integers modulo $d$.
Then the  two operators $Z$ and $X$ generate the  Heisenberg-Weyl group (up to phase factors), which reduces to the Pauli group in the case of a qubit.
The respective eigenbases of the three operators $Z$, $X$, and $XZ$  are mutually unbiased \cite{DurtEBZ10}, and a parsimonious strategy can be constructed using the three bases; here the test projectors can be expressed explicitly as in \eref{eq:MUBTP} below. So the maximally entangled state $|\Phi\>$ in any dimension $d$ can be verified within infidelity $\epsilon$ and significance level $\delta$ with only
\begin{equation}
\biggl\lceil\frac{\ln\delta}{\ln (1-2\epsilon/3)}\biggr\rceil\leq\Bigl\lceil\frac{3}{2\epsilon}\ln\delta^{-1}\Bigr\rceil
\end{equation}
 tests according to  \eref{eq:ParsiBoundN}. This number is only 50\% more than the number required by the best strategy based on entangling measurements.

 If $d$ has the prime-power decomposition $d=\prod_{j=1}^r p_j^{n_j}$, where $p_j$ are distinct primes and $n_j$ are positive integers,
then at least $\min_j (p_j^{n_j}+1)$ bases can be found that are mutually unbiased. In particular, 
a complete set of $d+1$ MUB can be constructed when the dimension is a prime power \cite{Ivan81,WootF89,DurtEBZ10}.  When the dimension $d$ is a  prime for example, the respective eigenbases of
$Z, X, XZ, XZ^2,..., XZ^{d-1}$ form a complete set of  MUB. In this case, the $d+1$ test projectors can be expressed as 
\begin{subequations}\label{eq:MUBTP}
\begin{align}
P_{0}&=\frac{1}{d}\sum_{k=0}^{d-1}(Z\otimes Z^{-1})^k=\sum_j |jj\>\<jj|,\\
P_{m+1}&=\frac{1}{d}\sum_{k=0}^{d-1} (XZ^m\otimes XZ^{-m})^k,\quad m=0,1,\ldots, d-1.
\end{align}
\end{subequations}
The resulting verification protocol is perfect, and the number of tests required to verify $|\Phi\>$ attains the lower bound in \eref{eq:OptimalBoundN}. When $d=2$, the three test projectors read
\begin{align}\label{eq:MUBqubitTP}
\frac{1+Z^{\otimes 2}}{2},\quad
\frac{1+X^{\otimes 2}}{2},\quad \frac{1-Y^{\otimes 2}}{2},\quad 
\end{align}
with $Y=\rmi XZ$, which reproduce the result in \rcite{PallLM18}.

When a complete set of MUB is not available, we can still devise optimal verification protocols for $|\Phi\>$ using (weighted complex projective) 2-designs. Let $P_+
$ be the projector onto the symmetric subspace of  $\caH^{\otimes 2}$. 
A weighted set of kets $\{|\psi_\xi\>, w_\xi\}$   in  $\caH$ with $w_\xi\geq0$ and  $\sum_\xi w_\xi=d$ is  a  \emph{2-design} \cite{Zaun11,ReneBSC04,Scot06} if 
$\sum_{\xi} w_\xi (|\psi_\xi\>\<\psi_\xi|)^{\otimes 2 }=2P_+/(d+1)$;
that is
\begin{equation}\label{eq:2designAlt}
\sum_{\xi} w_\xi (|\psi_\xi\>\<\psi_\xi|)\otimes (|\psi_\xi^*\>\<\psi_\xi^*|)
=\frac{1}{d+1}(1+d |\Phi\>\<\Phi|).
\end{equation}
Let $\{\caB_l, p_l\}_{l}$ be a weighted set of kets  with the uniform weight $p_l$ for all kets in basis $l$ and $\sum_lp_l=1$; note that the total weight is  $d\sum_l p_l =d$.   Then  $\{\caB_l, p_l\}_{l}$ forms a 2-design iff $\Omega=\sum_l p_l P(\caB_l)=(1+d |\Phi\>\<\Phi|)/(d+1)$. This observation confirms  the following result. 
\begin{proposition}\label{pro:CBoptimal}
A CB strategy  $\{P(\caB_l),p_l\}$ is optimal iff $\{\caB_l, p_l\}_{l}$ forms a 2-design.
\end{proposition}

\Psref{pro:CBparsimonious} and \ref{pro:CBoptimal} together imply the following result first derived in \rcite{RoyS07} (cf.~Theorem~3.2 there):
	at least $d+1$ bases are needed for constructing a 2-design in dimension $d$; if the lower bound is saturated, then all bases are mutually unbiased and have the same weight. When $d$ is a prime power, the lower bound can always be saturated \cite{WootF89,DurtEBZ10,KlapR05M,Zhu15M}. When $d+1$ is a prime power, a 2-design can be constructed from $d+2$ bases according to \rcite{RoyS07}, so an optimal CB strategy can be constructed using only $d+2$ measurement settings. When $d=6$ for example, a 2-design can be constructed from eight bases, although a complete set of MUB is not expected to exist.

\begin{proposition}\label{pro:OptCBTnum}
For any  maximally entangled state with local dimension
$d\geq 3$, an optimal verification protocol can be devised with at most $\lceil\frac{3}{4}(d-1)^2\rceil+1$ distinct CB tests.
\end{proposition}
\begin{proof}
According to Theorem~4.1 and Proposition~4.3 in \rcite{RoyS07}, a (weighted) 2-design can be constructed explicitly from $\lceil\frac{3}{4}(d-1)^2\rceil+1$ bases, from which we can devise an optimal verification protocol with $\lceil\frac{3}{4}(d-1)^2\rceil+1$ distinct CB tests by \pref{pro:CBoptimal}. 
\end{proof}

Thanks to \pref{pro:OptCBTnum},  the number of tests required to verify $|\Phi\>$ within infidelity $\epsilon$ and significance level $\delta$ can always attain the lower bound in \eref{eq:OptimalBoundN}. \Psref{pro:CBparsimonious} \ref{pro:CBoptimal}, and \ref{pro:OptCBTnum}  highlight the significance of MUB and 2-designs for the verification, fidelity estimation, and entanglement detection of the maximally entangled state $|\Phi\>$.

Next, we shall show that all parsimonious strategies based on LP measurements (that is, projective measurements on  product bases) and  all optimal strategies (including perfect strategies) based on SP measurements are actually CB strategies. These results further strengthen the significance of MUB and 2-designs. The following two theorems are proved in \aref{sec:OptStraProof}.
\begin{theorem}\label{thm:LPparsi}
An LP strategy  $\{P_l, p_l\}_{j=1}^g$ with $g\geq2$ is parsimonious iff  $p_l=1/g$ and $P_l=P(\caB_l)$, where the  $g$ bases $\caB_1,\caB_2, \ldots,\caB_g$  are mutually unbiased.  The strategy is perfect iff, in addition, $g=d+1$, so that $\caB_1, \caB_2,\ldots, \caB_g$ form a complete set of MUB.
\end{theorem}

\begin{remark}
	An LP strategy  is a  strategy 
	based on LP measurements, in which each test is realized by performing an LP measurement and then selecting suitable outcomes; cf.~\eref{eq:LPprojector} in the Appendix. By definition, each LP test projector
	on $\caH^{\otimes2}$ is diagonal in some product basis.
\end{remark}

\begin{theorem}\label{thm:Opt2design}
An SP strategy  $\{P_l, p_l\}_{j=1}^g$ is optimal iff there exist $g$ bases $\caB_1,\caB_2, \ldots,\caB_g$ such that $P_l=P(\caB_l)$ and $\{\caB_l,p_l \}_l$ forms a 2-design. 
The strategy is perfect iff, in addition, $g=d+1$, $p_l=1/(d+1)$, and $\caB_1,\caB_2, \ldots,\caB_g$ form a complete set of MUB. 
\end{theorem}
\Thref{thm:Opt2design} is quite surprising given that no obvious bases are involved in the definition of an SP strategy. In addition to the applications in  state verification, \thref{thm:Opt2design} also sheds light on the existence problem on MUB. 
% (and 2-designs)
\begin{corollary}
There exists a complete set of MUB in dimension $d$ iff there exist $d+1$ separable projectors $P_1, P_2,\ldots, P_{d+1}$ such that $P_j\geq |\Phi\>\<\Phi|$ for all $j$ and that $P_j-|\Phi\>\<\Phi|$ are mutually orthogonal.
\end{corollary}

\subsection{\label{sec:EntDetectNA}Applications to entanglement detection}
Note that  $\rho$ is entangled when the fidelity $F(\rho,|\Phi\>\<\Phi|)$ is larger than $1/d$. Given a verification strategy $\Omega$, to certify the entanglement of $|\Phi\>$ with significance level $\delta$,
the number of tests is given by \eref{eq:NumberTest} with $\epsilon=(d-1)/d$; that is,
\begin{equation}\label{eq:EDnumTest}
N_\rmE=\biggl\lceil\frac{ \ln \delta}{\ln[1-d^{-1}(d-1)\nu(\Omega)]}\biggr\rceil\geq \biggl\lceil\frac{\ln\delta^{-1}}{\ln d}\biggr\rceil.
\end{equation}
In particular, 
passing a single test can
guarantee the significance level 
$1-\frac{d-1}{d}\nu(\Omega)$
in the nonadversarial scenario, which also follows from  \eref{eq:OneTestCon}. The  bound in \eref{eq:EDnumTest} can be attained by the strategy $\Omega=|\Phi\>\<\Phi|$ based on the entangling measurement
$\{|\Phi\>\<\Phi|, 1-|\Phi\>\<\Phi|\}$.  In this case, the entanglement of $|\Phi\>$ can be certified using only one test for any given significance level $\delta$
if the local dimension satisfies the condition
 $d\geq \delta^{-1}$.

In practice, it is more convenient to apply strategies based on local measurements. 
If  $\Omega$ is the  optimal local strategy with $\nu(\Omega)=d/(d+1)$ (for example, the strategy based on a complete set of MUB), then \eref{eq:EDnumTest} reduces to 
\begin{equation}\label{eq:EDnumTestOpt}
N_\rmE=\biggl\lceil\frac{ \ln \delta}{\ln2-\ln(d+1)}\biggr\rceil.
\end{equation}
Surprisingly, the entanglement of $|\Phi\>$ can be certified with only one test based on local measurements  when
\begin{equation}\label{eq:OneTestLOCC}
d\geq 2\delta^{-1}-1,
\end{equation}
as illustrated in \fref{fig:ED}.
If $\Omega$ is a parsimonious strategy with $g$ distinct tests, then  $\nu(\Omega)=(g-1)/g$, so that \eref{eq:EDnumTest} reduces to 
\begin{equation}\label{eq:EDnumTestPasi}
N_\rmE=\biggl\lceil\frac{ \ln \delta}{\ln(g+d-1)-\ln(gd)}\biggr\rceil,
\end{equation}
which approaches $\lfloor\ln\delta^{-1}/\ln g\rfloor+1$
in the large-$d$ limit. Again,  one test is sufficient when $g>\delta^{-1}$ and $d$ is  large enough.

Recently, \rcite{DimiD18} (see also \rcite{SaggDGR19}) showed that the entanglement of certain multipartite states, such as linear cluster states and tensor powers of the singlet,  can be certified using only one test. In general it is not easy to make a fair comparison between their results and our results because the problems considered and starting points in the two works are different. With this caveat in mind, we present the following observations. In the case of the singlet, the protocol in  \rcite{DimiD18} is similar to the protocol in \rcite{PallLM18} and is a special case of our general protocol based on  complete sets of MUB [see \eref{eq:MUBqubitTP}].  Reference~\cite{DimiD18} essentially demonstrates that the entanglement of the singlet can be certified to any given significance if the number of copies is large enough, though the collection of singlets is considered as a whole and single-copy detection means a single copy of such a collection.  According to Eq.~(10) in \rcite{DimiD18}, to achieve significance level $\delta$ 
 (corresponding to confidence level $1-\delta$), the number of  required tests
satisfies
\begin{equation}
N\geq  \frac{\ln \delta}{\ln \frac{2}{3}}=\frac{\log_2 \delta}{\log_2 \frac{2}{3}},
\end{equation}
which corresponds to a local dimension of 
\begin{equation}
2^N\geq  \delta^{-1/\log_2(3/2)}\approx\delta^{-1.71}
\end{equation}
if these singlets are considered as a bipartite maximally entangled state. This dimension is in general much larger than 
the counterpart in \eref{eq:OneTestLOCC} required by our protocol. When $\delta=0.05$ for example,   the smallest local dimension required by  \rcite{DimiD18} is  $2^N=256$, while it is only 39 for our optimal protocol.  Our protocol is more efficient because it involves collective measurements across different copies if the maximally entangled state is composed of many copies of two-qubit Bell states.

\begin{figure}
	\includegraphics[width=6.7cm]{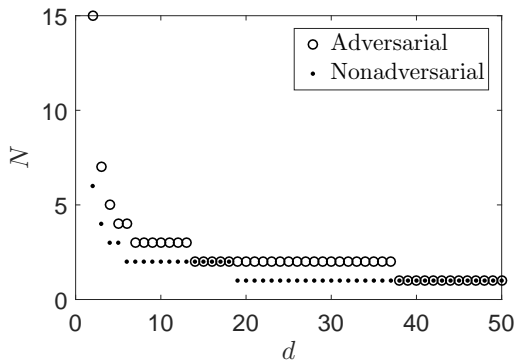}	
	\caption{\label{fig:ED}Certification of the entanglement of maximally entangled states in the adversarial scenario and nonadversarial scenario. Here $d$ is the local dimension and $N$ is the number of tests required to certify the entanglement with  significance level  $\delta=0.1$. In the nonadversarial scenario, the optimal strategy $\Omega$ with $\beta(\Omega)=1/(d+1)$ is applied [cf.~\eref{eq:OptOmega}]. In the adversarial scenario, a homogeneous strategy $\Omega'$ with $\beta(\Omega')=2/(d+1)$ is applied. In both scenarios,  one test is sufficient when $d$ is  large enough.}
\end{figure}

\section{Verification of maximally entangled states in the adversarial scenario}

\subsection{\label{sec:OVSadv}Optimal verification strategies}
In the adversarial scenario, the efficiency  of a verification strategy $\Omega$ will depend on smaller eigenvalues as well as $\beta(\Omega)$. In this case, singular verification strategies are not efficient for high-precision verification according to \esref{eq:NumberTestAdv} and \eqref{eq:NumberTestAdvEnt}, even for the strategy  $\Omega$ based on the  entangling test $\{|\Phi\>\<\Phi|, 1-|\Phi\>\<\Phi|\}$,
which is optimal for the nonadversarial scenario when there is no restriction on the measurements. Here we shall construct optimal protocols for the adversarial scenario.

If there is no restriction on the measurements, the optimal strategy can always be chosen to be homogeneous. In the  high-precision limit, a strategy $\Omega$ is optimal in the adversarial scenario if it is homogeneous with $\beta(\Omega)=1/\rme$ \cite{ZhuH19AdS, ZhuH19AdL}.  For the maximally entangled state $|\Phi\>$,  we can  construct a homogeneous strategy $\Omega$ with $\beta(\Omega)=1/(d+1)$ according to \sref{sec:OptNA}.
To construct the optimal strategy, it suffices to add the trivial test with a suitable probability $p$. Here  "trivial test" refers to the test associated with the identity operator, so that all states can pass the test for sure \cite{ZhuH19AdS, ZhuH19AdL}. Note that 
\begin{align}
|\Phi\>\<\Phi|+\lambda(1-|\Phi\>\<\Phi|)=(1-p)\frac{1+d|\Phi\>\<\Phi|}{d+1}+p
\end{align} 
if 
\begin{equation}
p=\frac{(d+1)\lambda-1}{d}.
\end{equation}
In this way,  any homogeneous strategy $\Omega$ that satisfies  $1/(d+1)\leq \beta(\Omega)< 1$  can be constructed by virtue of local projective measurements. In particular, we can construct a homogeneous strategy  $\Omega$ with $\beta(\Omega)=1/\rme$ by choosing 
$p=(d+1-\rme)/(\rme d)$.  Then the number of tests attains the minimum in the high-precision limit; that is,
\begin{equation}
N(\epsilon,\delta,\lambda)\approx\rme\epsilon^{-1}
\ln\delta^{-1}.
\end{equation}
In general, the  optimal value of $\beta(\Omega)$ depends on the target infidelity $\epsilon$ and significance level $\delta$. For high-precision verification, nevertheless, this  value is close to $1/\rme$ \cite{ZhuH19AdS, ZhuH19AdL}. Such strategies can also be constructed by virtue of local projective measurements.

\subsection{\label{sec:EntDetectAdv}Applications to entanglement detection}
Recall that a bipartite state $\rho$ is entangled whenever $\<\Phi|\rho|\Phi\>> 1/d$. 
Given a   strategy $\Omega$ for the maximally entangled state $|\Phi\>$, the number of tests required to certify its entanglement with significance level $\delta$ is  $N(\epsilon,\delta,\Omega)$ with $\epsilon=(d-1)/d$. Now the analysis in \sref{sec:OVSadv} is not so relevant because $\epsilon$ is quite large. Nevertheless, singular verification strategies are still not efficient when $\delta$ is small. 

If  $\Omega$ is the  strategy based on the entangling  test $\{|\Phi\>\<\Phi|, 1-|\Phi\>\<\Phi|\}$, then $\nu(\Omega)=1$, and the number of tests is given by \eref{eq:NumberTestAdvEnt} with $\epsilon=(d-1)/d$; that is, 
\begin{align}
N(\epsilon,\delta,\Omega)&=
\biggl\lceil\frac{d(1- \delta)}{(d-1)\delta }\biggr\rceil. \label{eq:NumberTestAdvEnt2}
\end{align}
When $\delta\ll1$, we have $N(\epsilon,\delta,\Omega)\approx d/[(d-1)\delta]$, so the number of tests is approximately inversely proportional to $\delta$. 
If  $\Omega$ is a parsimonious strategy composed of $g$ distinct tests with $g\leq d$ as constructed in \sref{sec:OptNA}, then 
the number of tests is determined by \eref{eq:NumberTestAdv} with $\nu(\Omega)=(g-1)/g$  and $\epsilon=(d-1)/d$; that is, 
\begin{align}
N(\epsilon,\delta,\Omega)=
\min\left\{
\biggl\lceil\frac{gd(1- \delta)}{(g-1)(d-1) \delta}\biggr\rceil,\; \biggl\lceil\frac{d}{(d-1)\delta}-1\biggr\rceil\right\}, 
\end{align}
which is approximately equal to the number in \eref{eq:NumberTestAdvEnt2}.

As in \sref{sec:OVSadv}, the optimal strategy for certifying the entanglement of $|\Phi\>$ can be chosen to be homogeneous. Given a homogeneous strategy $\Omega$  with $\beta(\Omega)=\lambda$,  the number of  required tests is $N(\epsilon,\delta,\lambda)$ presented in \eref{eq:NumTestHomo1} with $\epsilon=(d-1)/d$ \cite{ZhuH19AdL}. When $\delta\ll\lambda$, we have 
\begin{equation}
N(\epsilon,\delta,\lambda)\approx\frac{1+(d-1)\lambda}{(d-1)\lambda\ln\lambda}
\ln\delta
\end{equation}
according to \eref{eq:NumberTestAdvLS}.
The minimum of the right-hand side is attained when $\lambda$ is the unique solution, denoted by $\lambda_*$,  of the following equation
\begin{equation}
1+(d-1)\lambda+\ln \lambda=0. 
\end{equation}
It is not easy to derive an analytical formula for $\lambda_*$, but it is easy to compute $\lambda_*$ numerically. In addition, it is easy to prove that $\lambda_*\geq 1/(d+1)$ when $d\geq 4$, so the optimal strategy can be realized by LOCC. When $d=2,3$, the optimal value of $\lambda$ under LOCC is $1/(d+1)$. When $d\geq 3$ and $\lambda=1/(d-1)$, we have
\begin{equation}
N(\epsilon,\delta,\lambda)\approx\frac{2\ln \delta^{-1}}{\ln (d-1)},
\end{equation}
so the choice $\lambda=1/(d-1)$ is reasonably good for practical purposes. It should be pointed out that the above equation is derived under the assumption $\delta\ll \lambda=1/(d-1)$.

\begin{figure}
	\includegraphics[width=6.8cm]{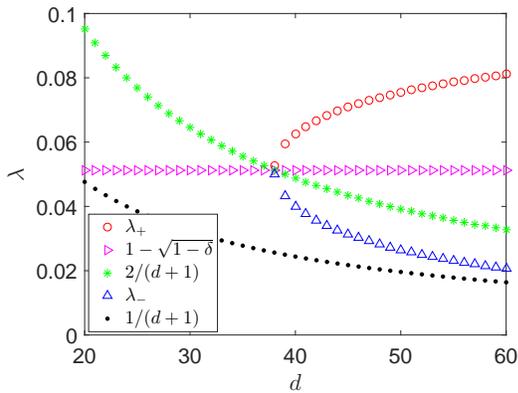}	
	\caption{\label{fig:EDoneTest}Entanglement certification of maximally entangled states  in the  adversarial scenario using only one test. Here the significance level is set at $\delta=0.1$. 
When the local dimension satisfies $d\geq38$, the entanglement can be certified by any homogeneous strategy $\Omega$ with $\lambda_-\leq \beta(\Omega)\leq \lambda_+$ [cf.~\esref{eq:betaCondition} and \eqref{eq:lambdapm}], including the specific choice $\beta(\Omega)=2/(d+1)$ or $\beta(\Omega)=1-\sqrt{1-\delta}$. 	The plot also shows that $\lambda_-\geq1/(d+1)$, which means  all such strategies can be realized by virtue of local projective measurements. 
}
\end{figure}

In the rest of this section we show that for any given significance level $0<\delta<1$, the entanglement of the maximally entangled state $|\Phi\>$ can be certified using only one test as long as the local dimension $d$ is large enough. To verify this claim, we may assume that  $0<\delta\leq 1/2$ without loss of generality because the number of tests cannot increase when $\delta$ increases. 
\begin{theorem}\label{thm:oneTestAdv}In the adversarial scenario,
	the entanglement of the $d\times d$ maximally entangled state $|\Phi\>$ can be certified with significance level $0<\delta\leq 1/2$ using only one test based on  a homogeneous strategy $\Omega$  iff 
	\begin{gather}
	d\geq d_*:= \biggl\lceil\frac{2+2\sqrt{1-\delta}-\delta}{\delta}\biggr\rceil, \label{eq:OneTestAdvME} \\
		\lambda_-\leq \beta(\Omega)\leq \lambda_+,\label{eq:betaCondition}
	\end{gather} 
where 
	\begin{equation}\label{eq:lambdapm}
	\lambda_\pm=\frac{(d+1)\delta\pm\sqrt{(d+1)^2\delta^2-4d\delta}}{2d}.
	\end{equation}
\end{theorem}

\begin{proof}
Let $\lambda=\beta(\Omega)$. If the entanglement of $|\Phi\>$ can be certified with significance level $0<\delta\leq 1/2$ using only one test, then  $\lambda>0$ 
according to \eref{eq:NumberTestAdvEnt2}; see also  the conclusion presented after \eref{eq:NumberTestAdvEnt}.  In addition,
 \eref{eq:OneTestConAdv} with $\epsilon=(d-1)/d$ has to hold, which means $\lambda<\delta$ and	
\begin{equation}\label{eq:dimlambound}
	d\geq \frac{\delta (1-\lambda)}{\lambda(\delta-\lambda)}.
	\end{equation}
The minimum of the right side is attained when $\lambda=1-\sqrt{1-\delta}$, in which case the above equation reduces to 
	\begin{equation}
	d\geq \frac{2+2\sqrt{1-\delta}-\delta}{\delta}, 
	\end{equation}
	which implies \eref{eq:OneTestAdvME}.
In addition, \eref{eq:dimlambound} implies that 
	\begin{align}
	d\lambda^2-(d+1)\delta\lambda+\delta\leq 0,
	\end{align}
	from which we can deduce that  $\lambda_-\leq \lambda=\beta(\Omega)\leq \lambda_+$, which confirms \eref{eq:betaCondition}. 
	
	Conversely, if  $d\geq d_*$ and $\lambda_-\leq \lambda=\beta(\Omega)\leq \lambda_+$, then \esref{eq:OneTestConAdv} and \eqref{eq:dimlambound} hold. 
	Therefore, the homogeneous strategy $\Omega$  can be applied to certify the entanglement of $|\Phi\>$ with only one test. 
\end{proof}
Note that  the bound in \eref{eq:OneTestAdvME} is equivalent to the condition $\delta\geq 4d/(d+1)^2$. When this condition is satisfied, then $\lambda_+$ ($\lambda_-$) is monotonically increasing  (decreasing) in $d$ and $\delta$. 
In conjunction with the assumption $0<\delta\leq 1/2$, we can deduce that 
\begin{align}
\frac{1}{d+1}&< \frac{1}{d}< \frac{d+1-\sqrt{d^2-6d+1}}{4d}\leq \lambda_-\nonumber\\
&\leq\lambda_+\leq \frac{d+1+\sqrt{d^2-6d+1}}{4d}< \frac{d-1}{2d},\label{eq:lambdapmBound} \\
&\frac{\delta}{d}<\lambda_-\leq \lambda_+<\delta. \label{eq:lambdapmBound2}
\end{align}
In addition, we have 
\begin{equation}\label{eq:lambdapmBound3}
\lambda_-\leq \frac{2}{d+1}\leq  1-\sqrt{1-\delta}\leq \lambda_+,
\end{equation}
where the middle inequality 
is saturated when the inequality $\delta\geq 4d/(d+1)^2$ is saturated, in which case all the inequalities in \eref{eq:lambdapmBound3} are saturated (cf.~\fref{fig:EDoneTest}).

If  \eref{eq:OneTestAdvME} is not  satisfied, then the entanglement of  $|\Phi\>$ cannot be certified  with only one test even if entangling measurements are accessible given that the  optimal performance can always be achieved by a  homogeneous strategy. 
Conversely, if \eref{eq:OneTestAdvME} is satisfied, then the entanglement of  $|\Phi\>$ can be certified  with only one test by a  homogeneous strategy constructed from local projective measurements. Actually, all homogeneous strategies that can certify the entanglement with one test can be realized by local measurements thanks to \eref{eq:lambdapmBound}. Notably, in view of  \eref{eq:lambdapmBound3},  the homogeneous strategy $\Omega$ with $\beta(\Omega)=1-\sqrt{1-\delta}$ or with $\beta(\Omega)=2/(d+1)$ can achieve the optimal performance and can be realized by local measurements.
For example, the entanglement of $|\Phi\>$ can be certified with significance level $\delta=0.1$ using one test iff  $d\geq 38$, as illustrated in  \fsref{fig:ED} and \ref{fig:EDoneTest}. When  $\delta\ll1$, we have $d_*\approx (4/\delta)-2$, so the threshold dimension is about two times the counterpart  $(2/\delta) -1$ for the nonadversarial scenario [cf.~\eref{eq:OneTestLOCC}].

\section{Summary}
We studied systematically efficient verification of maximally entangled states based on local projective measurements. 
We proved that  optimal strategies are in one-to-one correspondence with weighted complex projective 2-designs, while perfect strategies are in one-to-one correspondence with complete sets of MUB. Based on this observation, optimal protocols are constructed for maximally entangled states of any local dimension, and  near-optimal protocols are constructed using only three measurement settings. Besides state verification, these protocols are also very useful for fidelity estimation and entanglement detection. Moreover, our approach can be applied to the adversarial scenario. In this case, we can construct protocols based on local projective measurements that are optimal even among protocols that can access entangling measurements. In addition, we proved that the entanglement of the maximally entangled state can be certified with any given significance level using only one test when the local dimension is large enough.
Our work is of interest not only to practical quantum information processing, but also to foundational studies on the connections between quantum states and quantum measurements.

\bigskip

\bigskip
\section*{Acknowledgements}
HZ is grateful to Zihao Li and  Yun-Guang Han for comments.
This work is  supported by  the National Natural Science Foundation of China (Grant No. 11875110).
MH was supported in part by
Fund for the Promotion of Joint International Research (Fostering Joint International Research) Grant No. 15KK0007,
Japan Society for the Promotion of Science (JSPS) Grant-in-Aid for Scientific Research (A) No. 17H01280, (B) No. 16KT0017, and Kayamori Foundation of Informational Science Advancement.

\appendix
\section{\label{sec:OptStraProof}Proofs of \thref{thm:LPparsi} and \thref{thm:Opt2design}}
\begin{proof}[Proof of \thref{thm:LPparsi}]
	The "if" part of the theorem follows from \pref{pro:CBparsimonious}. Conversely, if the LP strategy  $\{P_l, p_l\}_{j=1}^g$ is parsimonious,  then  $p_l=1/g$  and $\bar{P}_l$ are mutually orthogonal according to \lref{lem:OptOrth}. In view of  \lref{lem:OrthMUB} below,  there must exist $g$ bases $\caB_1,\caB_2, \ldots,\caB_g$ that are mutually unbiased and such that $P_l=P(\caB_l)$. Therefore, the LP strategy is actually a CB strategy, which is perfect iff the bases form a complete set of MUB.
\end{proof}

\begin{proof}[Proof of \thref{thm:Opt2design}]
	The "if" part follows from  \psref{pro:CBoptimal} and \ref{pro:CBparsimonious}. Conversely, if the SP strategy  $\{P_l, p_l\}_{j=1}^g$  is optimal, then $\sum_l p_l\tr(P_l)=d$, which implies that all $P_l$ have rank $d$ by     \lref{lem:SepProjRank} below.
	So each $P_l$ has the form $P_l=P(\caB_l)$ for some  basis $\caB_l$ by  \lref{lem:Pbasis} below, and the SP strategy is actually a CB strategy. Now the theorem follows from  \psref{pro:CBoptimal} and \ref{pro:CBparsimonious}. 
\end{proof}

\begin{lemma}\label{lem:SepProjRank}
	Any separable projector $P\geq |\Phi\>\<\Phi|$ has rank at least $d$. 
\end{lemma}
\begin{proof}
	$\rank(P) \geq \rank (\tr_\rmB P)\geq \rank (\tr_\rmB |\Phi\>\<\Phi|)=d$. Alternatively, this conclusion follows from the observation that $\tr(P)\geq E_\caR(\Phi)+1=d$.
\end{proof}

\begin{lemma}\label{lem:OrthMUB}
	Suppose  $P_1, P_2\geq |\Phi\>\<\Phi|$ are two LP projectors. Then $\bar{P}_1$ and $\bar{P}_2$ are orthogonal iff $P_1=P(\caB_1)$ and $P_2=P(\caB_2)$, where $\caB_1$ and $\caB_2$ are two bases for $\caH$ that are mutually unbiased.
\end{lemma}

\begin{proof}
	The "if" part follows from \pref{pro:OrthMUB}.	Concerning the converse,	
	suppose $P_l$ is diagonal in the product basis $\caB_l\times \caB_l'$ for $l=1,2$. Then $P_l\geq P(\caB_l)$ by \lref{lem:LoProjOpt} below,  so  $\bar{P}_l\geq \bar{P}(\caB_l)$. If $\bar{P}_1$ and $\bar{P}_2$ are orthogonal, then $\bar{P}(\caB_1)$	and $\bar{P}(\caB_2)$ are orthogonal, so $\caB_1$ and $\caB_2$ are mutually unbiased by \pref{pro:OrthMUB}. In addition,
	\begin{equation}
	1=\tr(P_1 P_2)\geq \tr[P_1 P (\caB_2)]=\rank(P_1)/d,
	\end{equation}
	which implies that $P_1$ has rank $d$ in view of \lref{lem:SepProjRank}, so that $P_1=P(\caB_1)$. By the same token, $P_2=P(\caB_2)$.  
\end{proof}

\begin{lemma}\label{lem:LoProjOpt}
	Suppose  $P\geq |\Phi\>\<\Phi|$ is an LP  projector that is diagonal in the product basis $\caB\times \caB'$, then $P\geq P(\caB)$.
\end{lemma}

\begin{proof}
Suppose $\caB=\{|\psi_j\>\}_{j=1}^d$ and $\caB'=\{|\phi_j\>\}_{j=1}^d$; then $P$  has the form \begin{equation}\label{eq:LPprojector}
P=\sum_{j} \sum_{k\in A_j } |\psi_j\>\<\psi_j|\otimes|\phi_k\> \<\phi_k|,
\end{equation}
where $A_j$ are subsets of $\{1,2,\ldots, d\} $. In addition,
	\begin{align}
	&1=\<\Phi|P|\Phi\> =\sum_{j} \sum_{k\in A_j }  |\<\psi_j, \phi_k|\Phi\>|^2\nonumber\\
	& =\frac{1}{d}\sum_{j} \sum_{k\in A_j }  |\< \phi_k|\psi_j^*\>|^2\leq 1.
	\end{align}
	Here the upper bound is saturated iff each $|\psi_j^*\>$ is supported in  the span of $\{|\phi_k\>\}_{k\in A_j}$. It follows that  $|\psi_j ^*\>\<\psi_j^*|\leq  \sum_{k\in A_j } |\phi_k\> \<\phi_k|$, so that
	\begin{align}
	&P(\caB)=\sum_j |\psi_j\>\<\psi_j|\otimes |\psi_j ^*\>\<\psi_j^*|\nonumber\\
	&\leq \sum_{j} \sum_{k\in A_j } |\psi_j\>\<\psi_j|\otimes|\phi_k\> \<\phi_k|= P,
	\end{align}
which completes the proof.
\end{proof}

\begin{lemma}\label{lem:Pbasis}
	Any rank-$d$ separable projector $P$ on $\caH^{\otimes 2}$ that satisfies $P\geq|\Phi\>\<\Phi|$ has the form
	$P=P(\caB)$, where $\caB$ is an orthonormal basis for $\caH$.
\end{lemma}
\begin{remark}
The support of $P(\caB)$ contains exactly $d$ product states, namely, $|\psi\> \< \psi|  \otimes |\psi^*\>  \<\psi^*|$ for $|\psi\>\in \caB$. So the basis $\caB$ in \lref{lem:Pbasis} is uniquely determined by $P$, assuming  that two bases $\caB, \caB'$ for $\caH$ are deemed identical if they differ only by the ordering or  overall phase factors of kets [note that such bases yield the same test projector; that is, $P(\caB')=P(\caB)$].
Then the mapping from bases for $\caH$ to projectors on $\caH^{\otimes2}$ as defined by $P(\caB)$ in \eref{eq:CBproj} is injective.
\end{remark}

\begin{proof}[Proof of \lref{lem:Pbasis}]
	By assumption $P$ is a linear combination of pure product states with positive coefficients,
	\begin{equation}
	P=\sum_j c_j |\varphi_j\>\<\varphi_j|\otimes |\phi_j\>\<\phi_j|, 
	\end{equation}
	where  $|\varphi_j\>$ and $|\phi_j\>$ are normalized kets, $c_j>0$, and  $\sum_j c_j=\tr(P)=d$. 
	We have
	\begin{equation}
	1=\<\Phi| P|\Phi\>=\frac{1}{d}\sum_j c_j |\<\phi_j^*|\varphi_j\>|^2\leq \frac{1}{d}\sum_j c_j =1,
	\end{equation}
	which implies that $|\phi_j\>=|\varphi_j^*\>$ for all $j$.  So there exist $d$ kets, say $|\varphi_1\>,\ldots,  |\varphi_d\> $,  such that $|\varphi_j\>\otimes |\varphi_j^*\>$ for $j=1,2,\ldots, d$ span the support of $P$. In addition, $|\Phi\>$ has the form $|\Phi\> =\sum_{j=1}^d a_j |\varphi_j\>\otimes |\varphi_j^*\>$.
	This equality can hold iff $\caB:=\{|\varphi_j\>\}_{j=1}^d$ is an orthonormal basis for $\caH$ and  $a_j =1/\sqrt{d}$ for $j=1,2,\ldots,d$. Now $\{|\varphi_j\>\otimes |\varphi_j^*\>\}_{j=1}^d$ is an orthonormal basis in the support of $P$, so  $P=P(\caB)$.
\end{proof}

%%%%%%%%%%%%%%%%%%%%%%%%%%%%%%%%%%%%%%%%%%%%%%%%%%%%%%%%%%%%%%%%%%%%%%%%%%%%%%%%%%%%%%%%%%%%%%%%%%%%%%%%%%%%%%%%%%%%%%%%%%%%%%%%%%%%%%%%%%%%%%%%%%%%%%%%%%%%%%%%%%%%%%%%%%%%%%%%%%%%%%%%%%%%%%%%%%%%%%%%%%%%%%%%%%%%%%%%%%%%%%%%%%%%%%%%%%%%%%%%%%%%%%%%%%%%%%%%%%%%%%%%%%%%%%%%%%%%%%%%%%%%%%%%%%%%%%%%%%%

%@CONTROL{REVTEX41Control}
%@CONTROL{apsrev41Control,author="48",editor="1",pages="1",title="0",year="0"}

\nocite{apsrev41Control}
\bibliographystyle{apsrev4-1}
\bibliography{all_references}

\end{document}